\documentclass[letterpaper, 10 pt, conference]{ieeeconf} 

\usepackage{graphicx}
\usepackage{xpatch, amsmath,amssymb,bm,mathtools}
\usepackage{mleftright}
\usepackage{xcolor}
\usepackage{fp, tikz, pgfplots, multirow, nicefrac}
\usepackage{pgfplotstable}
\usepackage{filecontents}
\usepackage{multicol}
\usepackage{siunitx}
\usepackage{lipsum}
\usepackage{cite}
\usepackage[hidelinks]{hyperref}
\usepackage[capitalize]{cleveref}
\usepackage[labelfont=bf, font={footnotesize}]{caption}
\usepackage[pspdf={-dDELAYSAFER}]{auto-pst-pdf}
\usetikzlibrary{arrows}
\usetikzlibrary{shapes}
\usetikzlibrary{patterns}
\usetikzlibrary{decorations.pathmorphing}
\usetikzlibrary{decorations.pathreplacing}
\usetikzlibrary{decorations.shapes}
\usetikzlibrary{decorations.text}
\usetikzlibrary{shapes.misc}
\usetikzlibrary{decorations.markings}
\usetikzlibrary{arrows.meta}
\usetikzlibrary{backgrounds}
\usepgfplotslibrary{external}
\usepgfplotslibrary{fillbetween}
\usepgfplotslibrary{statistics}
\usetikzlibrary{positioning}
\usetikzlibrary{calc}

\pgfplotsset{compat=newest} 

\newtheorem{definition}{Definition}

\newtheorem{remark}{Remark}
\newtheorem{lemma}{Lemma}
\newtheorem{theorem}{Theorem}
\newtheorem{corollary}{Corollary}
\newtheorem{assumption}{Assumption}

\newcommand{\bmt}[1]{\bm{\tilde{#1}}}
\newcommand{\bmb}[1]{\bm{\bar{#1}}}
\newcommand{\bmh}[1]{\bm{\hat{#1}}}
\newcommand{\bmd}[1]{\bm{\dot{#1}}}
\newcommand{\bmdt}[1]{\bm{\dot{\tilde{#1}}}}

\newcommand{\R}{\mathbb{R}}
\newcommand{\tp}{^\mathsf{T}}

\newlength{\fwidth}
\newlength{\fheight}
\IEEEoverridecommandlockouts                              % This command is only needed if 
\title{\LARGE \bf
Is Data All That Matters? The Role of Control Frequency for Learning-Based Sampled-Data Control of Uncertain Systems
}

\author{Ralf R\"omer, Lukas Brunke, Siqi Zhou, and Angela P. Schoellig % <-this % stops a space
\thanks{The authors are with the Learning Systems and Robotics Lab (learnsyslab.org) and the Munich Institute for Robotics and Machine Intelligence, Technical University of Munich, Germany. Email: {\tt\footnotesize \{ralf.roemer; lukas.brunke; siqi.zhou; angela.schoellig\}@tum.de}.} %
}

\begin{document}

\maketitle
\thispagestyle{empty}
\pagestyle{empty}

%=============================================
\begin{abstract}    
    % Learning
Learning models or control policies from data has become a powerful tool to improve the performance of uncertain systems. While a strong focus has been placed on increasing the amount and quality of data to improve performance,  data can never fully eliminate uncertainty, making feedback necessary to ensure stability and performance. We show that the control frequency at which the input is recalculated is a crucial design parameter, yet it has hardly been considered before. 
    %%earning unknown system dynamics and control policies from data has become popular, with a strong focus being placed on increasing the amount and quality of data to improve performance.
    % Control frequency
    %%Nevertheless, uncertainty cannot be completely eliminated, making feedback necessary to ensure stability and performance.
    %Nevertheless, feedback is needed to compensate for remaining uncertainty and thereby ensure stability and performance.
    % Control Frequency: Often neglected
    %%The control frequency at which the input is recalculated is a crucial design parameter, yet its impact on stability and performance has hardly been considered.
    %The control frequency at which the input is recalculated is a crucial design parameter, but its impact has hardly been considered. %has hardly been considered in the context of learning-based control.
    % Combine GPs and sampled-data control
    We address this gap by combining probabilistic model learning and sampled-data control.
    % Contributions
    We use Gaussian processes (GPs) to learn a continuous-time model and compute a corresponding discrete-time controller. 
    The result is an uncertain sampled-data control system, for which we derive robust stability conditions. We formulate semidefinite programs to compute the minimum control frequency required for stability and to optimize performance.
    % Evaluation
    As a result, our approach enables us to study the effect of both control frequency and data on stability and closed-loop performance.
    We show in numerical simulations of a quadrotor that performance can be improved by increasing either the amount of data or the control frequency, and that we can trade off one for the other. 
    For example, by increasing the control frequency by~33\%, we can reduce the number of data points by half while still achieving similar performance.
    %Also, for a given uncertain system, we can improve performance by \textcolor{red}{44\%} by simply increasing the control frequency by \textcolor{red}{33\%} in our example.
    %The results demonstrate the potential of using the control frequency as a design parameter.
    %Moreover, we demonstrate that our approach allows us to reduce the control frequency when the GP model is updated while achieving at least similar performance.
    % Learning
    % Control frequency
    %The control frequency at which the input is recalculated is a crucial design parameter that has hardly been considered in learning-based control.
    % Learning has not considered frequency
    %Probabilistic learning methods such as Gaussian processes (GPs) can quantify the model uncertainty, which decreases as more data becomes available. However, in current learning-based control formulations, the model uncertainty cannot be related to the required minimum control frequency (MCF), as they model the system entirely in discrete- or continuous-time.  
\end{abstract}

\section{Introduction}
% Learning
Real-world systems such as robots can exhibit complex dynamics, making deriving accurate models from first principles difficult.
%For complex systems such as robots, it can be difficult to derive accurate dynamics models from first principles. 
Therefore, many studies 
in recent years % OPTIONAL
have addressed learning unknown dynamics from measured data using machine learning methods and designing a controller based on the learned model~\cite{Berkenkamp.2015, Ostafew.2016, Rohr.2021, Brunke.2022}. 
%Generally, the learned model cannot perfectly capture the system's dynamic behavior, for example, due to insufficient or corrupted data. 
Much attention has been paid to the role of data and increasing its amount and quality~\cite{Lederer.2021}. 
However, no derived or learned model can perfectly capture a real system's dynamic behavior~\cite{Nguyen.2011}. 
Therefore, feedback is required to guarantee stability and performance.
% Generally, no model, whether derived or learned, can perfectly capture a real system's dynamic behavior~\cite{Nguyen.2011}. 
% Although much attention has been paid to increasing the amount and quality of data~\cite{Lederer.2021}, feedback is still required to guarantee stability and performance.
% Control Frequency in learning-based control
The control frequency at which system measurements are fed back to recalculate the control input is often set without taking the dynamics and uncertainty into account~\cite{Berkenkamp.2015}, neglecting that it represents a degree of freedom in the controller design.
%The control frequency is often set without taking the dynamics and uncertainty into account~\cite{Berkenkamp.2015, Rohr.2021}, neglecting that it represents a degree of freedom in the controller design.
However, considering the control frequency as a design parameter can be advantageous, especially for systems such as resource-constrained robot platforms (e.g., drones).
For example, knowledge of the minimum control frequency (MCF) required for guaranteed stability of an uncertain system can help improve energy efficiency by reducing unnecessary computational demand and data transmission.
In this work, we study the effect of both control frequency and data on closed-loop performance.
%In this work, we investigate the relationship between control frequency and model uncertainty, affected by the amount of training data.

% Control Frequency in Learning-based Control
%\textcolor{red}{An important but rarely studied question is how frequently system measurements must be fed back to recalculate the control input. 
%The uncertainty of the learned model is generally coupled to the minimum control frequency (MCF) required for guaranteed stability.
%For practical systems, especially resource-constrained robot platforms, knowledge of the MCF can help reduce unnecessary computational demand and data transmission, thereby improving energy efficiency.
%}
%In this work, we %consider the control frequency as a design parameter and 
%investigate the impact of the control frequency on the stability and performance of a learning-based controller.
%the connection between the model uncertainty and the MCF.
% Therefore, we aim to relate the uncertainty of a learned dynamics model to the required control frequency to answer the question: How much feedback do we need in learning-based control?

% Relevance
%However, optimizing performance may not be the primary objective for systems with limited hardware resources. 
%Instead, saving energy, computational power and communication resources, e.g., to be used for other tasks such as environment perception, becomes crucial.
%A reduction of the control frequency can contribute to achieving these objectives. 

\begin{figure}
    \centering
    \small
    \begin{tikzpicture}
        % Blocks
        % \node [draw, minimum width=1cm, minimum height=0.6cm] (system) at (0,0) {Real system};
        % \node [draw, minimum width=1cm, minimum height=0.6cm, below right= 0.5cm and 1cm of system.center] (sampler) {Sampler};
        % \node [draw, minimum width=1cm, minimum height=0.6cm, below = 1cm of system.center] (controller) {Controller};
        % \node [draw, minimum width=1cm, minimum height=0.6cm, below left = 0.5cm and 1cm of system.center] (zoh) {ZOH};
        \draw[draw=black, dotted, fill=blue!10!white] (-1.8, -0.3) rectangle (1.95, -1.5);
        \node [draw, minimum width=1cm, minimum height=0.5cm] (system) at (0,0) {System};
        \node [draw, minimum width=1cm, minimum height=0.5cm, fill=white] (sampler) at ($(system.center) + (1.25cm,-0.6cm)$) {Sampler};
        \node [draw, minimum width=1cm, minimum height=0.5cm, fill=white] (controller) at ($(system.center) + (0cm,-1.2cm)$) {Controller};
        \node [draw, minimum width=1cm, minimum height=0.5cm, fill=white] (zoh) at ($(system.center) + (-1.25cm,-0.6cm)$) {ZOH};
        \node [draw, align=center, text width=3.1cm, minimum height=0.4cm, dashed, fill=red!10!white] (gp) at ($(controller.center) + (0cm,-1cm)$) {GP Learning \footnotesize{(Sec.~\ref{subsec:meth_learning})}};
        \node[] (data) at ($(gp.north) + (0cm,0.25cm)$) {Data};
        %\\ Sec.~\ref{subsec:meth_learning} - \ref{subsec:meth_uncertainty_reparameterization}};
        \node [draw, align=center, text width=2.3cm, minimum height=0.4cm, dashed, fill=red!10!white] (sd) at (-5.2, -1) {Robust Sampled-Data \\ Control Design \\ \footnotesize{(Sec.~\ref{subsec:meth_sd_stability} \& \ref{subsec:meth_sd_performance})}}; %\\ };    
        %\draw[draw=black, dotted, fill=black!5!white] ($(zoh.north west) + (-0.05cm,0.05cm)$) rectangle ($(sampler.south east) + (0.05cm,-0.8cm)$);
        \node (fc) at (1.75, -1.3) {$f_\mathrm{c}$};    

        % Arrows
        \draw[-stealth] (system.east) -| (sampler.north); %node[midway,above] {$\bm{x}(t)$};
        \draw[-stealth] (sampler.south) |- (controller.east); %node[midway,below] {$\bm{x}(t_k)$};
        \draw[-stealth] (controller.west) -| (zoh.south); %node[midway,below] {$\bm{u}(t_k)$};
        \draw[-stealth] (zoh.north) |- (system.west);
        % node[midway,above] {$\bm{u}(t)$};
        \draw[-stealth, dashed] ($(controller.center) + (-1cm,0cm)$) |- ($(gp.north) + (-0.5cm,0.3cm)$) -- ($(gp.north) + (-0.5cm,0cm)$);
        \draw[-stealth, dashed] ($(controller.center) + (1cm,0cm)$) |- ($(gp.north) + (0.5cm,0.3cm)$) -- ($(gp.north) + (0.5cm,0cm)$);
        \draw[-stealth, dashed] (gp.west) -| (sd.south) node[pos=0.35,below] {Uncertain Continuous-Time Model};
        \draw[-stealth, very thick] ($(sd.east) + (0.25cm,0.25cm)$) -- ($(sd.east) + (1.9,0.25cm)$) node[midway,above] {\begin{tabular}{c}
            Control \\ Frequency $f_\mathrm{c}$
        \end{tabular}}; 
        \draw[-stealth, very thick] ($(sd.east) + (0.25cm,-0.1cm)$) -- ($(sd.east) + (1.9,-0.1cm)$) node[midway,below] {\begin{tabular}{c}
             Controller \\ Parameters
        \end{tabular}};
    \end{tikzpicture}
    \caption{In digital control systems (blue shaded box), the sampler, controller, and zero-order-hold~(ZOH) operate at a certain control frequency~$f_\mathrm{c}$. 
    We propose a framework (dashed-line boxes) to simultaneously compute the minimum control frequency and design a controller using an uncertain model learned from data using Gaussian process (GP) regression.
    %In this work, we simultaneously design a controller using measured data and Gaussian process (GP) regression and compute the minimum control frequency for robust stabilization of the uncertain system.
    }
    \label{fig:intro_fig}
\end{figure}
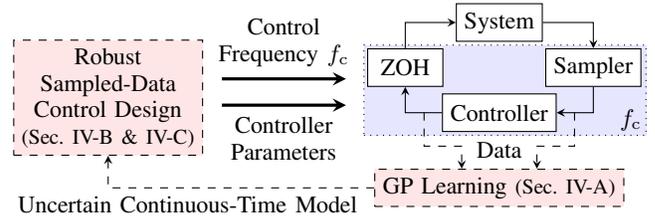

% Relevant concepts: GPs, sampled-data system
To quantify the uncertainty inherent in a learned dynamics model, probabilistic methods such as Bayesian linear regression~(BLR)~\cite{Bishop.2006} or Gaussian processes~(GPs)~\cite{Rasmussen.2006} have become popular~\cite{Brunke.2022}. While BLR assumes linearity in a set of parameters, GPs have the advantage of being a non-parametric method.
% due to their flexibility in modeling nonlinear systems and ability to adaptively incorporate new data in real time. 
In~\cite{Berkenkamp.2015, Ostafew.2016, Rohr.2021}, GPs are combined with robust control methods.
%The combination of GP model learning with robust control methods is studied, e.g., in \cite{Berkenkamp.2015, Ostafew.2016, Rohr.2021}. 
These works demonstrate that control performance is impacted by model uncertainty and, thus, the training data, but they do not investigate the role of the control frequency.

% Related work
%Our work takes inspiration from \cite{Berkenkamp.2015}, where a classical robust controller is synthesized for a linearized GP dynamics model.
Few studies have considered the control frequency in the context of model uncertainty.
%In~\cite{Berberich.2021}, a direct data-driven approach for computing the maximum sampling interval for stabilizing an unknown linear system is proposed.
In~\cite{Berberich.2021}, the maximum sampling interval for stabilizing an unknown linear system is computed directly from data. 
However, the approach assumes bounded noise and involves a computationally expensive iterative optimization scheme for controller design. %whose computational complexity increases \todo{how?} with the amount of data and is therefore unsuitable for higher-dimensional systems such as robots.
GP-based feedback linearization with a data-dependent delay for updating the control input is considered in~\cite{Dai.2023}.
%discrete-time controller update and a data-dependent computation time is considered.
It is shown empirically that in terms of tracking accuracy, there can be a tradeoff between the accuracy of the GP model and the computational delay.
In~\cite{Metelli.2020}, reducing the control frequency is discussed in the context of reinforcement learning as a way to reduce the sample complexity and thus improve performance.

Learning a discrete-time model of unknown dynamics yields a model specific to a particular sampling time. %providing robust stability guarantees 
Analyzing the system's behavior and stability
for a different controller sampling time is generally very difficult.
Thus, we consider learning a continuous-time model in this work. Designing a discrete-time controller for a continuous-time system falls within the domain of sampled-data control.
%Different approaches have been proposed to analyze the stability of sampled-data systems \cite{Zhang. 2023}. 
%The stability of sampled-data systems can be analyzed, and stability conditions can be derived via the time-delay approach~\cite{Fridman.2004, Fridman.2010, Seuret.2012}. 
The stability of sampled-data systems can be analyzed via the time-delay approach~\cite{Fridman.2014}, and stability conditions are derived in~\cite{Fridman.2004, Fridman.2010, Seuret.2012}. %~\cite{Fridman. 2014}, which uses Lyapunov-Krasovskii functionals to derive constructive stability conditions as matrix inequalities~\cite{Fridman.2004, Fridman.2010, Seuret.2012}. 
%While these conditions can guarantee robust stability for polytopic-type uncertainty, they do not apply to the uncertainty of a learned dynamics model.
These conditions can also guarantee robust stability for polytopic-type uncertainty, but this becomes computationally intractable for the uncertainty set associated with a dynamics model learned from data.
% The input-delay approach is the most popular Among the different approaches to sampled-data control with non-uniform sampling. This approach derives stability conditions for sampled-data systems using Lyapunov-Krasovskii functionals. 

% Contribution
We propose a framework to design the control frequency based on the uncertainty associated with a dynamics model learned from data as illustrated in Fig.~\ref{fig:intro_fig} and study the role of the control frequency compared to the amount of data. Taking a sampled-data control approach, we robustly stabilize a partially unknown nonlinear continuous-time system with a discrete-time controller. 
%For this, we learn the residual continuous-time dynamics using GP regression and linearize the learned model around a desired operating point. A reparameterization of the GP model uncertainty gives a sampled-data system with norm-bounded uncertainty in the system matrices. 
%For this uncertain system, we derive robust stability conditions as matrix inequalities. Based on these, we formulate the computation of the MCF and the performance optimization as semidefinite programs (SDPs).
%Through numerical simulations of a quadrotor, we show that our approach allows us to study the tradeoff between the amount of data and the control frequency required for stability. We observe that if the model uncertainty decreases, the control frequency can be reduced without sacrificing stability or performance.
Our main contributions are:
\begin{itemize}
    % Highlight the general novelty of combining GPs/probabilistic model learning and sampled-data control
    \item We combine GP-based stochastic model learning with sampled-data control to study the effect of the control frequency and model uncertainty on the closed-loop performance.
    % Technical contributions
    \item We derive robust stability conditions as matrix inequalities for a sampled-data control system with learned uncertain dynamics. Based on these, we formulate semidefinite programs (SDPs) for the computation of the MCF and performance optimization. 
    Our framework enables us to control the system at different frequencies without having to re-learn the model.
    % Evaluation/general key insights 
    \item Through numerical simulations\footnote{All code for reproducing the results reported in this paper is available at \url{https://github.com/ralfroemer99/lb_sd}}, we show and analyze the tradeoff between model uncertainty, affected by the amount of data collected, and control frequency in terms of stability and performance.
\end{itemize}

% Summarize benefits 
Our results demonstrate that the choice of the control frequency can be as crucial as collecting more data (to further reduce uncertainty). 

% Paper Structure
%The remainder of this work is structured as follows: Section~\ref{sec:prob} states the considered problem. Section~\ref{sec:background} introduces the technical background on GPs and sampled-data systems. Our proposed learning-based sampled-data control approach is presented in Section~\ref{sec:methodology} and evaluated numerically in Section~\ref{sec:eval}. We discuss the results in Section~\ref{sec:disc} before drawing conclusions in Section~\ref{sec:conc}.

\textbf{Notation}:
We denote the Kronecker product by~$\otimes$, the Hadamard (element-wise) product by~$\circ$ and the probability by~$\mathrm{Pr}(\cdot)$.
In symmetric matrices, $*$ denotes transpose elements that can be inferred from symmetry.
Given a square matrix ${\bm{A}\in \R^{n \times n}}$, ${\mathrm{diag}(\bm{A}) \in \R^n}$ is a vector containing the diagonal elements of $\bm{A}$. 
Given a vector ${\bm{a}\in \R^n}$, ${\mathrm{Diag}(\bm{a})\in \R^{n \times n}}$ is a diagonal matrix containing the elements of $\bm{a}$ on its diagonal. 
Given a matrix~${\bm{B} = \big[\bm{b}_1\, \dots\, \bm{b}_m\big]\in \R^{n \times m}}$, we denote its vectorization by~${\mathrm{vec}(\bm{B}) = \big[\bm{b}_1\tp\, \dots\, \bm{b}_m\tp\big]\tp \in \R^{nm}}$.

\section{Problem Statement} \label{sec:prob}
We consider a dynamical system whose state and input at time $t \in \R_{\geq 0}$ are given by $\bm{x}(t)\in \R^n$ and $\bm{u}(t)\in \R^m$, respectively. The system evolves according to the dynamics %the partially known continuous-time dynamics
\begin{align}
    \label{eq:prob_open_loop_system}
    \bmd{x}(t) = \bm{h}(\bm{x}(t), \bm{u}(t)) = \bm{f}(\bm{x}(t), \bm{u}(t)) + \bm{g}(\bm{x}(t), \bm{u}(t)),\!
\end{align}
where the function $\bm{f}:\R^n \times \R^m \rightarrow \R^n$ is known, for example, derived from first principles, and the function ${\bm{g}:\R^n \times \R^m \rightarrow \R^n}$ is unknown, accounting for unmodeled dynamic effects.
Both $\bm{f}$ and $\bm{g}$ are assumed to be continuously differentiable. We denote~${\bm{z}=(\bm{x},\bm{u}) \in \R^{n_z}}$, where $n_z=n+m$, for brevity.
We assume the availability of noisy measurement data collected from system~\eqref{eq:prob_open_loop_system}.
\begin{assumption}
    A dataset of $N$ observations from~\eqref{eq:prob_open_loop_system}
    \begin{align}
        \label{prob:eq_dataset_assumption}
        \mathcal{D} = \left\{\bm{z}^{(i)}, \bm{y}^{(i)}=\bmd{x}^{(i)}-\bm{f}\big(\bm{z}^{(i)}\big)+\bm{w}^{(i)}\right\}_{i=1}^N
    \end{align}
    is available, with targets perturbed by independent and identically distributed (i.i.d.) Gaussian noise ${\bm{w}^{(i)}\sim \mathcal{N}(\bm{0}, \bm{\Sigma}_\mathrm{n})}$.
\end{assumption}

While this assumption requires exact state and input measurements, it allows for Gaussian perturbed observations of the state derivative, which is often approximated via finite differences in practice. Similar assumptions are made, for example, in~\cite{Lederer.2021, Dai.2023}.

The samples in $\mathcal{D}$ serve as training data to learn an approximation of $\bm{g}$, denoted by $\bmh{g}$.
% \begin{assumption}
%     The components of $\bm{g}(\cdot)\!=\!\left[g_1(\cdot),\dots,g_n(\cdot)\right]\tp$ are drawn from zero-mean GPs, i.e., $g_i(\cdot) \sim \mathcal{GP}(0,k_i(\cdot,\cdot))$.
% \end{assumption}
We aim to control system~\eqref{eq:prob_open_loop_system} robustly around a known equilibrium~$\bm{z}_\mathrm{e}=(\bm{x}_\mathrm{e}, \bm{u}_\mathrm{e})$ with $\bm{0} = \bm{f}(\bm{z}_\mathrm{e}) + \bm{g}(\bm{z}_\mathrm{e})$ despite the uncertainty associated to~$\bmh{g}$. 
%Assuming knowledge of~$\bm{z}_\mathrm{e}$ 
This does not represent a significant restriction as we can estimate any unknown equilibrium~$\bm{z}_\mathrm{e}$ by solving a nonlinear optimization problem involving $\bm{f}$ and $\bmh{g}$, cf.~\cite{Berkenkamp.2015}.
We use a discrete-time controller with zero-order-hold
\begin{align}
    \label{eq:prob_input}
    \bm{u}(t) = \bm{\pi}(\bm{x}(t_k)), \qquad \forall t \in [t_k, t_{k+1}),
\end{align}
where $\bm{\pi}(\cdot)$ is a control law based on the learned continuous-time model, and $t_k$, $k\in\! \mathbb{N}_0$, ${t_0 = 0}$, are the sampling instants. %, on which we make the following assumption.
%Combining~\eqref{eq:prob_open_loop_system} and ~\eqref{eq:prob_input} yields a sampled-data control system.
We make an assumption on the sampling instants capturing periodic and aperiodic sampling and define the MCF.
\begin{assumption} \label{ass:sd_sampling_interval}
    The interval between two consecutive sampling instants satisfies ${t_{k+1} - t_k \leq T_\mathrm{s}}$, $\forall k\in \mathbb{N}_0$, where $T_\mathrm{s}$ is an upper bound on the sampling interval.
\end{assumption}

%This assumption captures standard periodic sampling as well as aperiodic sampling. We define the MCF as follows.
\begin{definition} \label{def:mcf}
    Let $T_{\mathrm{s,max}}$ be the largest value of $T_\mathrm{s}$ such that the system~\eqref{eq:prob_open_loop_system} with the control~\eqref{eq:prob_input} satisfying~\cref{ass:sd_sampling_interval} can be robustly stabilized around~$\bm{z}_\mathrm{e}$. Then, the MCF is given by ${f_\mathrm{c,min} = \frac{1}{T_\mathrm{s,max}}}$.
    %Given~\cref{ass:sd_sampling_interval}, the minimum control frequency (MCF) is defined as ${f_\mathrm{c,min} = \frac{1}{T_\mathrm{s,max}}}$, where $T_{\mathrm{s,max}}$ is the largest value of $T_\mathrm{s}$ such that the system~\eqref{eq:prob_open_loop_system} with the discrete-time control~\eqref{eq:prob_input} can be robustly stabilized around~$\bm{z}_\mathrm{e}$.
    % satisfying~\eqref{eq:prob_operating_point}.
\end{definition}

The MCF depends on the employed feedback control law~$\bm{\pi}(\cdot)$.
We consider the problem of learning a local linear approximation of $\bm{g}$ from the data~\eqref{prob:eq_dataset_assumption}, computing the MCF together with a linear control law based on the uncertainty of the learned dynamics and aim to optimize the control performance.
Ultimately, we want to study the impact of the control frequency on stability and performance, especially compared to model uncertainty and the amount of data.

%for a given value of $T_\mathrm{s} \in (0,T_{\mathrm{s,max}}]$.
%Given our emphasis on online learning control, one main objective is to devise a computationally efficient solution.

\section{Background} \label{sec:background}

\subsection{Gaussian Process Regression} \label{subsec:background_gp}
To simplify notation, we consider learning an unknown scalar function ${g:\R^{n_z}\rightarrow \R}$ from training inputs~$\bm{z}^{(i)}$ and target~${y^{(i)}=g(\bm{z}^{(i)})+w^{(i)}}$, ${i=1,\dots, N}$, which are perturbed by i.i.d. noise ${w^{(i)}\sim \mathcal{N}(0,\sigma_\mathrm{n}^2)}$.
GP regression~\cite{Rasmussen.2006} assumes that the unknown function is drawn from a GP, denoted as $\mathcal{GP}(\mu(\cdot),k(\cdot,\cdot))$, which induces a distribution over functions
%A GP induces a distribution $\mathcal{GP}(\mu(\cdot),k(\cdot,\cdot))$ over functions, 
such that any finite number of function evaluations is jointly Gaussian distributed.
The prior mean function ${\mu:\R^{n_z}\rightarrow \R}$ can incorporate prior knowledge in the form of an approximate model, and the kernel ${k:\R^{n_z}\times\R^{n_z} \rightarrow \R}$ encodes information about the structure of the unknown function. 
Without loss of generality, we set the mean function to zero. 
%, as prior information about the dynamics is captured by~$\bm{f}$.
%, as we only aim to learn the residual dynamics $\bm{g}(\cdot)$ in~\eqref{eq:prob_open_loop_system}. 
Under the GP assumption, the vector of observed targets $\bm{y}=\big[y^{(1)},\dots,y^{(N)}\big]\tp$ and the function value at a query point $\bm{z}^*\in\R^{n_z}$ have the joint probability distribution
\begin{align}
    \label{gp:joint_pd}
    \begin{bmatrix}
        \bm{y} \\ g(\bm{z}^*)
    \end{bmatrix} \sim
    \mathcal{N} \left(\bm{0}, \begin{bmatrix}
        \bmb{K} & \bm{k}(\bm{z}^*)\tp \\
        % \bm{k}(\bm{z}^*) 
        * & k(\bm{z}^*,\bm{z}^*)
    \end{bmatrix} \right),
\end{align}
where the gram matrix~$\bmb{K}$ and vector~$\bm{k}(\bm{z}^*)$ are defined as ${\bmb{K}=\bm{K} + \sigma_\mathrm{n}^2 \bm{I}_n}$, where ${\bm{K}_{ij}=k\left(\bm{z}^{(i)},\bm{z}^{(j)}\right)}$, and ${\bm{k}(\bm{z}^*)\! =\! \big[k\left(\bm{z}^{(1)},\bm{z}^*\right),\dots,k\left(\bm{z}^{(N)},\bm{z}^*\right)\big]\tp}$, respectively. Conditioning $g(\bm{z}^*)$ on the training data yields the posterior predictive distribution ${g(\bm{z}^*) \sim \mathcal{N}(\mu(\bm{z}^*),\sigma^2(\bm{z}^*))}$ with mean ${\mu(\bm{z}^*) = \bm{k}(\bm{z}^*)\tp \bmb{K}^{-1} \bm{y}}$ and variance ${\sigma^2(\bm{z}^*) = k(\bm{z}^*,\bm{z}^*) - \bm{k}(\bm{z}^*)\tp\bmb{K}^{-1} \bm{k}(\bm{z}^*)}$.
% \begin{align}
%     \label{eq:gp_posterior_mean}
%     \mu(\bm{z}^*) &= \bm{k}(\bm{z}^*)\tp \bmb{K}^{-1} \bm{y}, \\
%     \label{eq:gp_posterior_variance}
%     \sigma^2(\bm{z}^*) &= k(\bm{z}^*,\bm{z}^*) - \bm{k}(\bm{z}^*)\tp\bmb{K}^{-1} \bm{k}(\bm{z}^*).
% \end{align}

As the derivative is a linear operator, the derivative of a GP is also a GP \cite{Rasmussen.2006}. We can use this property to predict the derivative of $g$ at $\bm{z}^*$, denoted~${\left.\frac{\partial g(\bm{z})}{\partial \bm{z}}\right|_{\bm{z}^*}}$. From~\eqref{gp:joint_pd}, we obtain
\begin{align}
    \begin{bmatrix}
        \bm{y} \\ \left.\frac{\partial g(\bm{z})}{\partial \bm{z}}\right|_{\bm{z}^*}
    \end{bmatrix} \sim
    \mathcal{N} \left(\bm{0}, \begin{bmatrix}
        \bmb{K} & \left.\frac{\partial \bm{k}(\bm{z})}{\partial \bm{z}}\right|_{\bm{z}^*} \\
        % \left.\frac{\partial \bm{k}(\bm{z})\tp}{\partial \bm{z}}\right|_{\bm{z}^*} 
        * & \left.\frac{\partial^2 k(\bm{z},\bm{z})}{\partial \bm{z} \partial \bm{z}} \right|_{\bm{z}^*}
    \end{bmatrix} \right).
\end{align}
Similar to the derivation of $g(\bm{z}^*)$ from~\eqref{gp:joint_pd}, conditioning ${\left.\frac{\partial g(\bm{z})}{\partial \bm{z}}\right|_{\bm{z}^*}}$ on the observations~$\bm{y}$ yields the predictive distribution
%predictions of the derivatives can be obtained via conditioning as 
${\left.\frac{\partial g(\bm{z})}{\partial \bm{z}}\right|_{\bm{z}^*} \sim \mathcal{N}(\bm{\mu}'(\bm{z}^*),\bm{\Sigma}'(\bm{z}^*))}$ with mean and variance
\begin{align}
    \label{eq:gp_mean_derivative}
    \bm{\mu}'(\bm{z}^*)\! &= \!
    \left. \frac{\partial \bm{k}(\bm{z})}{\partial \bm{z}} \right|_{\bm{z}^*} \!\!\! \bm{\bar{K}}^{-1} \bm{y}, \\
    \label{eq:gp_variance_derivative}
    \bm{\Sigma}'(\bm{z}^*)\! &= \!
    \left.\frac{\partial^2 k(\bm{z},\bm{z})}{\partial \bm{z} \partial \bm{z}} \right|_{\bm{z}^*} 
    \!\!-\! \left.\frac{\partial \bm{k}(\bm{z})}{\partial \bm{z}}\right|_{\bm{z}^*} \!\!\!\bm{\bar{K}}^{-1} \!\left(\!\left.\frac{\partial \bm{k}(\bm{z})}{\partial \bm{z}}\right|_{\bm{z}^*} \right)\tp\!\!.\!
\end{align}

\subsection{Sampled-Data Systems} \label{subsec:background_sd}
Consider a continuous-time LTI system
\begin{align}
    \label{eq:sd_nominal_open_loop_system}
    \bm{\dot{x}}(t) = \bm{A}\bm{x}(t) + \bm{B}\bm{u}(t), \qquad \bm{x}(0) = \bm{x}_0,
\end{align}
under a discrete-time linear state-feedback control law
\begin{align}
    \label{eq:sd_control_law_lin}
    \bm{u}(t) = \bm{K}\bm{x}(t_k), \qquad \forall t \in [t_k, t_{k+1}),
\end{align}
where $\bm{K}\in \R^{m \times n}$, and the sampling instants $t_k$ satisfy~\cref{ass:sd_sampling_interval}. 
%Stability of the resulting closed-loop system can be analyzed from different perspectives~\cite{Zhang.2023}. 
The time-delay approach~\cite{Fridman.2004, Fridman.2010, Seuret.2012, Fridman.2014} to sampled-data systems writes the closed-loop system as
\begin{align}
    \label{eq:sd_nominal_closed_loop_system}
    \bm{\dot{x}}(t) = \bm{A}\bm{x}(t) + \bm{B}\bm{K}\bm{x}(t-\tau(t)), 
\end{align}
where ${\tau(t) = t-t_k}$, $\forall t \in [t_k, t_{k+1})$, is a piecewise-continuous, time-varying delay. %which is bounded by ${\tau(t)\in[0,T_\mathrm{s}]}$, $\forall t$, as a result of Assumption~\ref{ass:sd_sampling_interval}. 
Stability of~\eqref{eq:sd_nominal_closed_loop_system} can be analyzed using the Lyapunov-Krasovskii functional~\cite{Fridman.2004}
\begin{align}
    \label{eq:sd_lyapunov_functional}
    V(t) = \bm{x}(t)\tp \bm{P}_1 \bm{x}(t) + \!\int_{-T_\mathrm{s}}^0 \int_{t+\theta}^t \bmd{x}(s)\tp \bm{P}_2 \bmd{x}(s) \mathrm{d}s \mathrm{d} \theta,
\end{align}
where $\bm{P}_1 \succ \bm{0}$, $\bm{P}_2 \succ \bm{0}$. The second term in~\eqref{eq:sd_lyapunov_functional} is used to compensate for the delay-dependent part in $\frac{\mathrm{d}}{\mathrm{d}t}\bm{x}(t)\tp \bm{P}_1 \bm{x}(t)$.

\begin{lemma} [\cite{Fridman.2004}, Lemma 2.3] \label{lem:sd_nominal_step2}
    The control~\eqref{eq:sd_control_law_lin} asympto-tically stabilizes system~\eqref{eq:sd_nominal_open_loop_system} for all samplings satisfying Assumption~\ref{ass:sd_sampling_interval} if there exist matrices $\bm{Q}_1 = \bm{Q}_1\tp \succ \bm{0}$, $\bm{Q}_2$, $\bm{Q}_3$, $\bm{Z}_1$, $\bm{Z}_2$, $\bm{Z}_3$, $\bm{R} = \bm{R}\tp \succ \bm{0}$, all in $\R^{n \times n}$, and ${\bm{Y} \in \R^{m \times n}}$, that satisfy the matrix inequalities
    \begin{align}
        \label{eq:sd_nominal_lmi3}
        \bm{W}_{\bm{A},\bm{B}} \triangleq 
        \begin{bmatrix}
            \bm{\Xi} & \bm{\Xi}_{\bm{A}, \bm{B}} & T_\mathrm{s}\bm{Q}_2\tp \\
            * & -\bm{Q}_3 - \bm{Q}_3\tp + T_\mathrm{s}\bm{Z}_3 & T_\mathrm{s}\bm{Q}_3\tp \\
            * & * & -T_\mathrm{s}\bm{R} 
        \end{bmatrix} \prec \bm{0}, \\
        \label{eq:sd_nominal_lmi4}
        \begin{bmatrix}
            \bm{Q}_1 \bm{R}^{-1}\bm{Q}_1 & \bm{0} & \bm{Y}\tp \bm{B}\tp \\
            * & \bm{Z}_1 & \bm{Z}_2 \\
            * & * & \bm{Z}_3
        \end{bmatrix} \succeq \bm{0},
    \end{align}
    where $*$ denotes symmetry, ${\bm{\Xi} = \bm{Q}_2 + \bm{Q}_2\tp + T_\mathrm{s} \bm{Z}_1}$ and $\bm{\Xi}_{\bm{A}, \bm{B}} = \bm{Q}_3 - \bm{Q}_2\tp + \bm{Q}_1\bm{A}\tp + T_\mathrm{s}\bm{Z}_2 + \bm{Y}\tp \bm{B}\tp$. The stabilizing state-feedback gain is then given by $\bm{K}=\bm{Y}\bm{Q}_1^{-1}$. 
\end{lemma}
\begin{proof}
    The proof involves multiple steps, including defining a descriptor system of~\eqref{eq:sd_nominal_closed_loop_system}, differentiating~\eqref{eq:sd_lyapunov_functional} for time and using that the delay is bounded by ${\tau(t)\in[0, T_\mathrm{s}]}$ due to~\cref{ass:sd_sampling_interval}. See~\cite{Fridman.2004, Fridman.2014} for details.
\end{proof}

\section{Methodology} \label{sec:methodology}
In this section, we first discuss %present our approach for learning-based robust sampled-data control of system~\eqref{eq:prob_open_loop_system}. 
%use the properties of GPs introduced in Section~\ref{subsec:background_gp} to obtain a 
learning a probabilistic estimate of the dynamics~\eqref{eq:prob_open_loop_system} and a locally valid uncertain linearization from the dataset~\eqref{prob:eq_dataset_assumption}. Then, we present our approach for robust sampled-data control for different control frequencies based on the learned model's uncertainty.

\subsection{Model Learning and Linearization} \label{subsec:meth_learning}
%We make the following assumption on the unknown function~$\bm{g}$ in~\eqref{eq:prob_open_loop_system}.
We train a GP model for each output dimension of ${\bm{g}(\cdot)=\left[g_1(\cdot),\dots,g_n(\cdot)\right]\tp}$, assuming the following:
\begin{assumption}
    \label{ass:gp_se_kernel}
    The functions $g_i$, $i=1,\dots,n$, are drawn from zero-mean GPs with squared-exponential (SE) kernel
    %, i.e., ${g_i(\cdot)\sim \mathcal{GP}(0,k_i(\cdot,\cdot))}$, where
    \begin{align}
        \label{eq:gp_se_kernel}
        k_i\left(\bm{z},\bm{z}'\right) &= \sigma_{\eta,i}^2 \exp{\left(-\frac{1}{2}(\bm{z}-\bm{z}')\tp \bm{L}_i^{-2}(\bm{z}-\bm{z}')\right)},
    \end{align}
    where $\sigma_{\eta,i}^2>0$ and ${\bm{L}_i=\mathrm{Diag}(\bm{l}_i) \in \R^{n \times n}}$, $\bm{l}_i>\bm{0} \in \R^n$.
\end{assumption}

%For this, we assume that 
%This approach greatly reduces the computational complexity; however, note that it cannot capture possible correlations between the output dimensions of $\bm{g}$. 
%The learned approximation of $\bm{g}$ is given by $\bmh{g}(\cdot) = \big[\mu_1(\cdot),\dots, \mu_n(\cdot)\big]\tp$.

In~\eqref{eq:gp_se_kernel}, $\sigma_{\eta,i}^2$ is the output variance, and $\bm{L}_i$ contains the vector of length scales~$\bm{l}_i$, which corresponds to the rate of change of~$g_i$ with respect to~$\bm{z}$. 
\cref{ass:gp_se_kernel} is not restrictive in practice as the corresponding space of sample functions of the GP contains all continuous functions \cite{Lederer.2019}. 
The kernel hyperparameters are often unknown, but they can be determined, for example, by maximizing the marginal log-likelihood of the training data~\eqref{prob:eq_dataset_assumption}; see~\cite{Rasmussen.2006}.
We compute the derivative of the $i$-th GP at $\bm{z}_\mathrm{e}$, denoted by ${\left.\frac{\partial g_i(\bm{z})}{\partial \bm{z}}\right|_{\bm{z}_\mathrm{e}}\sim \mathcal{N}\left(\bm{\mu}_i'(\bm{z}_\mathrm{e}),\bm{\Sigma}_i'(\bm{z}_\mathrm{e})\right)}$, via~\eqref{eq:gp_mean_derivative} and~\eqref{eq:gp_variance_derivative}, where the partial derivatives of~\eqref{eq:gp_se_kernel} are straightforward to evaluate. This gives a probabilistic estimate of the linearized dynamics.
\begin{lemma}
    \label{lem:linearized_dynamics}
    Under Assumption~\ref{ass:gp_se_kernel} and given the data~\eqref{prob:eq_dataset_assumption}, the linearization of~\eqref{eq:prob_open_loop_system} about $\bm{z}_\mathrm{e}$ satisfies that for any $p \in [0,1)$, 
    \begin{align}
        \mathrm{Pr}\left(\left.\frac{\bm{h}}{\partial \bm{z}}\right|_{\bm{z}_\mathrm{e}} 
        \in \mathcal{C} \right) \geq p^n,
    \end{align}
    where $\mathcal{C} = \left[\bmh{C} - \bmb{C}, \bmh{C} + \bmb{C}\right] \subset \R^{n \times n_z}$, where
    \begin{align}
        \label{eq:gp_linearized_model_nominal}
        \bmh{C} &= \begin{bmatrix}
            \bmh{A} & \bmh{B} 
            \end{bmatrix} = 
            \left.\frac{\partial \bm{f}(\bm{z})}{\partial \bm{z}} \right|_{\bm{z}_\mathrm{e}} + 
        \begin{bmatrix}
            \bm{\mu}_1'(\bm{z}_\mathrm{e})\tp \\ \vdots \\ \bm{\mu}_n'(\bm{z}_\mathrm{e})\tp
        \end{bmatrix}, \\
        \label{eq:gp_linearized_model_uncertainty}
        \bmb{C} &= \begin{bmatrix} 
            \bmb{A} & \bmb{B} 
            \end{bmatrix} =
            \sqrt{\chi_{n_z}^2(p)}
        \begin{bmatrix}
            \sqrt{\mathrm{diag}(\bm{\Sigma}_1'(\bm{z}_\mathrm{e}))\tp} \\ \vdots \\ \sqrt{\mathrm{diag}(\bm{\Sigma}_n'(\bm{z}_\mathrm{e}))\tp}
        \end{bmatrix}.
    \end{align}
    Here, $\chi_{n_z}^2$ is the quantile function of the chi-squared distribution of degree~$n_z$.
\end{lemma}
\begin{proof}
    The fundamental properties of multivariate Gaussian distributions imply that for all  
    $p\in(0,1]$, ${\mathrm{Pr} \left(\!\left. \left. \frac{\partial g_i(\bm{z})}{\partial \bm{z}} \right|_{\bm{z}_\mathrm{e}}\! - \bm{\mu}_i'(\bm{z}_\mathrm{e}) \in \mathcal{E}_i \right|\mathcal{D}\right) = p}$, where $\mathcal{E}_i$ is an ellipsoidal confidence region defined by ${\mathcal{E}_i = \left\{\bm{d} \in \R^{n_z}\left|\bm{d}\tp \left(\bm{\Sigma}_i'(\bm{z}_\mathrm{e}) \right)^{-1} \bm{d} \leq \chi_{n_z}^2(p) \right. \right\}}$, ${i=1,\dots,n}$. %Here, $\chi_{n_z}^2$ is the quantile function of the chi-squared distribution with~$n_z$ degrees of freedom. 
    The result follows from the independence of the GPs and the fact that $\mathcal{E}_i \subseteq \mathcal{B}_i$, where $\mathcal{B}_i$ is a hyperrectangle with dimensions ${2\sqrt{\chi_{n_z}^2\!(p) \left(\bm{\Sigma}_i'(\bm{z}_\mathrm{e})\right)_{jj}}}$, ${j=1,\dots,n_z}$, which is symmetric about the origin.
\end{proof}

%We denote ${\begin{bmatrix} \bm{A} & \bm{B} \end{bmatrix} = \bm{\Theta}}$, ${\begin{bmatrix} \bmh{A} & \bmh{B} \end{bmatrix} = \bmh{\Theta}}$, where ${\bm{A}, \bmh{A} \in \R^{n \times n}}$. 
Due to~\cref{lem:linearized_dynamics}, the true linearized dynamics at $\bm{z}=\bm{z}_\mathrm{e}$ are captured with probability of at least~$p^n$ by 
\begin{align}
    \label{eq:gp_uncertain_system1}
    \bm{\dot{\tilde{x}}}(t) = \big(\bmh{A} + \bmb{A} \circ \bm{\Omega}\big) \bmt{x}(t) +
    \big(\bmh{B} + \bmb{B} \circ \bm{\Psi}\big) \bmt{u}(t),
\end{align}
where $\bmt{x}=\bm{x}-\bm{x}_\mathrm{e}$, $\bmt{u}=\bm{u}-\bm{u}_\mathrm{e}$ are deviations about the equilibrium, and ${\bm{\Omega} \in [-1,1]^{n \times n}}$, ${\bm{\Psi} \in [-1,1]^{n \times m}}$ are unknown. 

As $\mathcal{C}$ is a polytopic set, \cref{lem:sd_nominal_step2} can be applied to analyze the stability of~\eqref{eq:gp_uncertain_system1} in principle. 
However, this requires evaluating the stability conditions~\eqref{eq:sd_nominal_lmi3} and~\eqref{eq:sd_nominal_lmi4} for all $2^{nn_z}$ vertices of $\mathcal{C}$, which is computationally infeasible except for very low-dimensional systems. The same holds for the stability conditions in~\cite{Fridman.2010, Seuret.2012}. To address this problem, we reparameterize the Hadamard product terms in~\eqref{eq:gp_uncertain_system1} corresponding to the uncertainty by making use of the following lemma, which is straightforward to prove.
%\subsection{Uncertainty Reparameterization} %\label{subsec:meth_uncertainty_reparameterization}
\begin{lemma} \label{lem:unc_reparameterization_trick}
    Let ${\bm{U} = \big[\bm{u}_1\, \dots\, \bm{u}_n\big]\tp, \bm{V} \in \R^{n \times m}}$. Then, %the Hadamard product $\bm{U} \circ \bm{V}$ can expressed as 
    \begin{align*}
        \bm{U} \circ \bm{V} = (\bm{I}_n \otimes \bm{1}_{1\times m}) \mathrm{Diag}\left(\mathrm{vec}\left(\bm{V}\tp\right)\right) \begin{bmatrix}
            \mathrm{Diag}(\bm{u}_1) \\ \vdots \\ \mathrm{Diag}(\bm{u}_n)
        \end{bmatrix}. 
    \end{align*}
    %where $\otimes$ denotes the Kronecker product, and ${\mathrm{vec}(\bm{V}\tp) \in \R^{nm}}$ is the vectorization of $\bm{V}\tp$.
\end{lemma}

We denote the rows of $\bmb{A}$ and $\bmb{B}$ by $\bmb{a}_1\tp,\dots,\bmb{a}_n\tp$ and $\bmb{b}_1\tp,\dots,\bmb{b}_n\tp$, respectively, and define~${p = n^2 + nm}$. Then, we can use~\cref{lem:unc_reparameterization_trick} to rewrite the uncertainty in~\eqref{eq:gp_uncertain_system1} as
\begin{align}
    \label{eq:gp_uncertainty_reparameterization}
    \bmb{A} \circ \bm{\Omega} = \bm{H}\bm{\Delta}\bm{E}, \qquad \bmb{B} \circ \bm{\Psi} = \bm{H}\bm{\Delta}\bm{F},
\end{align}
where ${\bm{\Delta}=\mathrm{Diag}([\delta_1,\dots,\delta_p]) \in \R^{p \times p}}$ with $|\delta_i|\leq 1$, $\forall i$, $\bm{H} = \big[\bm{I}_n \otimes \bm{1}_{1 \times n}, \bm{I}_n \otimes \bm{1}_{1 \times m}]$ and
\begin{align}
\begin{split}
    \bm{E} &= \begin{bmatrix}
        \mathrm{Diag}\big(\bmb{a}_1\big) & \dots & \mathrm{Diag}\big(\bmb{a}_n\big) & \bm{0}_{n \times nm}
    \end{bmatrix}\tp, \\
    \bm{F} &= \begin{bmatrix}
        \bm{0}_{n \times n^2} & \mathrm{Diag}\big(\bmb{b}_1\big) & \dots & \mathrm{Diag}\big(\bmb{b}_n\big)
    \end{bmatrix}\tp.
\end{split}
\end{align}
Inserting~\eqref{eq:gp_uncertainty_reparameterization} into~\eqref{eq:gp_uncertain_system1} gives the uncertain linearized system
\begin{align}
    \label{eq:sd_norm_bounded_open_loop_system}
    \bm{\dot{\tilde{x}}}(t) = \big(\bmh{A} + \bm{H}\bm{\Delta}\bm{E}\big)\bmt{x}(t) + \big(\bmh{B} + \bm{H}\bm{\Delta}\bm{F}\big)\bmt{u}(t),
\end{align}
where all matrices except for $\bm{\Delta}$ are known. %Next, we discuss robust sampled-data control of the uncertain system~\eqref{eq:sd_norm_bounded_open_loop_system}.
%Note that the unknown matrix $\bm{\Delta} \in \R^{p \times p}$ is an unknown matrix satisfying ${\bm{\Delta}\tp \bm{\Delta} \preceq \bm{I}}$, and $\bm{H} \in \R^{n \times p}, \bm{E} \in \R^{p \times n}, \bm{F} \in \R^{p \times m}$ are known. 
%\begin{remark}
    %The accuracy of the uncertainty estimates obtained from GP regression depends on the parameterization of the kernel.
%\end{remark}

\subsection{Robust Sampled-Data Control of the Uncertain System} \label{subsec:meth_sd_stability}
To robustly control the continuous-time system~\eqref{eq:sd_norm_bounded_open_loop_system}, we consider a discrete-time linear state feedback 
\begin{align}
    \label{eq:sd_control_law_linearized}
    \bmt{u}(t) = \bm{K}\bmt{x}(t_k), \qquad \forall t \in [t_k, t_{k+1}),
\end{align}
where the sampling instants $t_k$, $k\in \mathbb{N}_0$, 
satisfy~\cref{ass:sd_sampling_interval}. 
%The resulting closed-loop system \textcolor{red}{write out} can be written with a time-varying delay similar to~\eqref{eq:sd_nominal_closed_loop_system}.
This results in the uncertain closed-loop system
\begin{align}
    \label{eq:sd_norm_bounded_closed_loop_system}
    \!\!\bmdt{x}(t)\! =\! \big(\bmh{A}\! +\! \bm{H}\bm{\Delta}\bm{E}\big)\bmt{x}(t)\! + \!\big(\bmh{B}\! +\! \bm{H}\bm{\Delta}\bm{F}\big)\bm{K}\bmt{x}(t\! -\! \tau(t)),\!\!
\end{align}
where the delay ${\tau(t)}$ is defined similar to~\cref{subsec:background_sd} and, thus, bounded by ${\tau(t)\in[0,T_\mathrm{s}]}$. 
%However, \cref{lem:sd_nominal_step2} and the conditions stated in~\cite{Fridman.2010, Seuret.2012} are
%still applicable if $\bm{A}$ and $\bm{B}$ belong to a known polytope, but
%not applicable to~\eqref{eq:sd_norm_bounded_open_loop_system}. 
We employ the following result to take the norm-bounded uncertainty in the linearized GP dynamics model into account:
\begin{lemma}[\cite{Xie.1996}]
%, Lemma 2.4] 
\label{lem:sd_norm_bounded_trick}
    Let~${\bm{\Theta}\in \R^{m \times m}}$ satisfy ${\bm{\Theta}\tp \bm{\Theta} \preceq \bm{I}}$. Then, for all constant matrices ${\bm{U} \in \R^{n \times m}}$, ${\bm{V} \in \R^{m \times n}}$ and all scalars~${\epsilon > 0}$, it holds that
    \begin{align*}
    \begin{split}
        -\epsilon^{-1}\bm{U}\bm{U}\tp - \epsilon \bm{V}\tp \bm{V} 
        &\preceq \bm{U}\bm{\Theta} \bm{V} + \bm{V}\tp \bm{\Theta}\tp\bm{U}\tp \\
        &\preceq \epsilon^{-1}\bm{U}\bm{U}\tp + \epsilon \bm{V}\tp \bm{V}.
    \end{split}
    \end{align*}
\end{lemma}

Using~\cref{lem:sd_norm_bounded_trick}, we derive the following constructive conditions for robust stability of the closed-loop system.
\begin{theorem} \label{the:sd_norm_bounded_step2}
    The uncertain system~\eqref{eq:sd_norm_bounded_closed_loop_system} is robustly asymptotically stable for all samplings satisfying Assumption~\ref{ass:sd_sampling_interval} if there exist matrices ${\bm{Q}_1 = \bm{Q}_1\tp \succ \bm{0}}$, $\bm{Q}_2$, $\bm{Q}_3$, $\bm{Z}_1$, $\bm{Z}_2$, $\bm{Z}_3$, $\bm{R} = \bm{R}\tp \succ \bm{0}$, all in~$\R^{n \times n}$, ${\bm{Y} \in \R^{m \times n}}$ and scalars $\epsilon_1 > 0$, $\epsilon_2 > 0$, that satisfy the matrix inequalities
    \begin{align}
        \label{eq:sd_norm_bounded_lmi3}
        \mleft[
        \begin{array}{c c c | c c}
            \multicolumn{3}{c}{\multirow{3}{*}{\large $\bm{W}_{\bmh{A},\bmh{B}}$}} 
            \vline & \bm{0} & \epsilon_1\left(\bm{Q}_1\bm{E}\tp + \bm{Y}\tp\bm{F}\tp\right) \\
            \multicolumn{3}{c}{} \vline & \bm{H} & \bm{0} \\
            \multicolumn{3}{c}{} \vline & \bm{0} & \bm{0} \\
            \hline
            * & * & * & -\epsilon_1 \bm{I} & \bm{0} \\
            * & * & * & * & -\epsilon_1 \bm{I}
        \end{array}
        \mright] 
        &\prec \bm{0}, \\
        \label{eq:sd_norm_bounded_lmi4}
        \begin{bmatrix}
            2\bm{Q}_1 - \bm{R} & \bm{0} & \bm{Y}\tp \bmh{B}\tp & \bm{0} & \epsilon_2 \bm{Y}\tp \bm{F}\tp \\
            * & \bm{Z}_1 & \bm{Z}_2 & \bm{0} & \bm{0} \\
            * & * & \bm{Z}_3 & \bm{H} & \bm{0} \\
            * & * & * &  \epsilon_2 \bm{I} & \bm{0} \\
            * & * & * & * & \epsilon_2 \bm{I}
        \end{bmatrix} 
        &\succeq \bm{0},
    \end{align}
    where $\bm{W}_{\bmh{A},\bmh{B}}$ is defined in~\eqref{eq:sd_nominal_lmi3}. The stabilizing state-feedback gain is then given by $\bm{K}=\bm{Y}\bm{Q}_1^{-1}$. 
\end{theorem}

\begin{proof}
    The idea is to show that if~\eqref{eq:sd_norm_bounded_lmi3} and~\eqref{eq:sd_norm_bounded_lmi4} are satisfied, then the nominal stability conditions~\eqref{eq:sd_nominal_lmi3} and~\eqref{eq:sd_nominal_lmi4} hold for all realizations of the uncertain system matrices~${\bm{A} = \bmh{A} + \bm{H} \bm{\Delta} \bm{E}}$ and~${\bm{B} = \bmh{B} + \bm{H} \bm{\Delta} \bm{F}}$ of~\eqref{eq:sd_norm_bounded_closed_loop_system}. We start by applying the Schur complement~\cite{Boyd.1994} to the first inequality~\eqref{eq:sd_norm_bounded_lmi3} in~\cref{the:sd_norm_bounded_step2}, which gives
    \begin{align}
        \label{eq:proof2_1}
        \bm{W}_{\bmh{A},\bmh{B}} + 
        \begin{bmatrix} 
            \bm{M} & \epsilon_1\bm{N}\tp 
        \end{bmatrix}
        \begin{bmatrix}
            \epsilon_1^{-1} \bm{I} & \bm{0} \\ 
            \bm{0} & \epsilon_1^{-1} \bm{I}
        \end{bmatrix}
        \begin{bmatrix}
            \bm{M}\tp \\
            \epsilon_1\bm{N}
        \end{bmatrix} &\prec \bm{0},
    \end{align}
    where~${\bm{M} = \big[\bm{0} \;\; \bm{H}\tp \;\; \bm{0} \big]\tp}$ and ${\bm{N} = \big[ \bm{E}\bm{Q}_1\tp + \bm{F} \bm{Y} \;\; \bm{0} \;\; \bm{0} \big]}$. By rewriting~\eqref{eq:proof2_1} and making use of~\cref{lem:sd_norm_bounded_trick}, we obtain
    \begin{align}
        \bm{W}_{\bmh{A},\bmh{B}} + \epsilon_1^{-1} \bm{M} \bm{M}\tp + \epsilon_1 \bm{N}\tp \bm{N} &\prec \bm{0} \nonumber \\
        \bm{W}_{\bmh{A},\bmh{B}} + \bm{M} \bm{\Delta} \bm{N} + \bm{N}\tp \bm{\Delta}\tp \bm{M}\tp &\prec \bm{0}, 
        \label{eq:proof2_2}
    \end{align}
    for all~${\bm{\Delta} \in \R^{p \times p}}$ satisfying~$\bm{\Delta}\tp \bm{\Delta} \preceq \bm{I}$.
    Inserting the expressions for~$\bm{M}$, $\bm{N}$ and~$\bm{W}_{\bmh{A},\bmh{B}}$ into~\eqref{eq:proof2_2} yields
    \setlength{\arraycolsep}{2.7pt}
    \begin{align}
        \underbrace{
        \bm{W}_{\bmh{A},\bmh{B}} + \!
        \begin{bmatrix}
            \bm{0} & \bm{Q}_1 (\bm{H} \bm{\Delta} \bm{E})\tp + \bm{Y}\tp (\bm{H} \bm{\Delta} \bm{F})\tp & \bm{0} \\
            * & \bm{0} & \bm{0} \\
            * & * & \bm{0}
        \end{bmatrix}}_{=\bm{W}_{\bm{A}, \bm{B}}} \! &\prec \bm{0}, \!
        % \nonumber \\
        % \begin{bmatrix}
        %         \bm{\Xi} & \bm{\Xi}_{\bm{A}, \bm{B}}
        %         %+ \bm{Q}_1(\bm{H}\bm{\Delta}\bm{E})\tp + \bm{Y}\tp(\bm{H}\bm{\Delta}\bm{F})\tp 
        %         & T_\mathrm{s}\bm{Q}_2\tp \\
        %         * & -\bm{Q}_3 - \bm{Q}_3\tp + T_\mathrm{s}\bm{Z}_3 & T_\mathrm{s}\bm{Q}_3\tp \\
        %         * & * & -T_\mathrm{s}\bm{R} 
        % \end{bmatrix} \! &\prec \bm{0} \nonumber \\
        % \bm{W}_{\bm{A},\bm{B}} &\prec \bm{0}. \!
        \label{eq:proof2_3}
    \setlength{\arraycolsep}{3pt}
    \end{align}
    which is equivalent to~\eqref{eq:sd_nominal_lmi3}.
    Thus, the satisfaction of~\eqref{eq:sd_norm_bounded_lmi3} implies the satisfaction of~\eqref{eq:sd_nominal_lmi3} for all realizations~${\bm{A} = \bmh{A} + \bm{H} \bm{\Delta} \bm{E}}$, ${\bm{B} = \bmh{B} + \bm{H} \bm{\Delta} \bm{F}}$ of the uncertain system matrices. 
    Consider now the second inequality~\eqref{eq:sd_norm_bounded_lmi4} in~\cref{the:sd_norm_bounded_step2}. As a result of $\bm{R}=\bm{R}\tp\succ \bm{0}$ and~$\bm{Q}_1=\bm{Q}_1\tp$, we have
    \begin{gather}
        (\bm{Q}_1 - \bm{R})\tp\bm{R}^{-1}(\bm{Q}_1 - \bm{R}) = \bm{Q}_1\bm{R}^{-1}\bm{Q}_1 - 2\bm{Q}_1 + \bm{R} \succ \bm{0} \nonumber \\
        \label{eq:proof1_2}
        \Rightarrow \bm{Q}_1\bm{R}^{-1}\bm{Q}_1 \succ 2\bm{Q}_1 - \bm{R}.
    \end{gather}
    Hence, the satisfaction of~\eqref{eq:sd_norm_bounded_lmi4} implies
    \begin{align}
        \label{eq:sd_norm_bounded_lmi4_bilinear}
        \begin{bmatrix}
            \bm{Q}_1 \bm{R}^{-1}\bm{Q}_1 & \bm{0} & \bm{Y}\tp \bmh{B}\tp & \bm{0} & \epsilon_2 \bm{Y}\tp \bm{F}\tp \\
            * & \bm{Z}_1 & \bm{Z}_2 & \bm{0} & \bm{0} \\
            * & * & \bm{Z}_3 & \bm{H} & \bm{0} \\
            * & * & * &  \epsilon_2 \bm{I} & \bm{0} \\
            * & * & * & * & \epsilon_2 \bm{I}
        \end{bmatrix} &\succeq \bm{0}.
    \end{align}
    Applying the Schur complement to~\eqref{eq:sd_norm_bounded_lmi4_bilinear} and performing similar steps as above results in
    \begin{align}
        \begin{bmatrix}
            \bm{Q}_1 \bm{R}^{-1}\bm{Q}_1 & \bm{0} & \bm{Y}\tp (\bmh{B} + \bm{H} \bm{\Delta} \bm{F})\tp \\
            * & \bm{Z}_1 & \bm{Z}_2 \\
            * & * & \bm{Z}_3
        \end{bmatrix} \succeq \bm{0},
    \end{align}
    for all~$\bm{\Delta}\tp \bm{\Delta} \preceq \bm{I}$, which is equivalent to~\eqref{eq:sd_nominal_lmi4} with ${\bm{B} = \bmh{B} + \bm{H}\bm{\Delta}\bm{F}}$. Thus, if~\eqref{eq:sd_norm_bounded_lmi4} holds, then~\eqref{eq:sd_nominal_lmi4} is satisfied for all realizations of~$\bm{B}$,  which concludes the proof.
\end{proof}

\begin{remark}
    We could also assume ${\bm{R} = \epsilon_3 \bm{Q}_1}$ for some~${\epsilon_3 > 0}$ to convexify the problem~\cite{Fridman.2004}.
    % As an alternative to~\eqref{eq:proof1_2}, we could assume ${\bm{R} = \epsilon_3 \bm{Q}_1}$ for some~${\epsilon_3 > 0}$ to convexify the problem as in~\cite{Fridman.2004}. 
    However, this would add another scalar decision variable to~\cref{the:sd_norm_bounded_step2}, increasing the computational complexity.
\end{remark}

\cref{the:sd_norm_bounded_step2} enables us to analyze the stability of the uncertain system~\eqref{eq:sd_norm_bounded_open_loop_system} for a 
%\textit{given} control~\eqref{eq:sd_control_law_linearized} and a 
\textit{given} upper bound on the sampling interval~$T_\mathrm{s}$. Moreover, we can compute the MCF~${f_\mathrm{c,min} = \frac{1}{T_\mathrm{s,max}}}$ and the corresponding stabilizing control gain~${\bm{K}=\bm{Y}\bm{Q}_1^{-1}}$
%\begin{corollary} \label{cor:mcf_computation}
    %The MCF for system~\eqref{eq:sd_norm_bounded_open_loop_system} under the control~\eqref{eq:sd_control_law_linearized} can be obtained 
by solving the optimization problem
\begin{align}
\begin{split}
\label{eq:sd_norm_bounded_stability_op}
    \min_{\substack{
            T_\mathrm{s}, \bm{Q}_1, \bm{Q}_2, \bm{Q}_3, \bm{Z}_1, \\
            \bm{Z}_2, \bm{Z}_3, \bm{R}, \bm{Y}, \epsilon_1, \epsilon_2
        }}
    \quad & \frac{1}{T_\mathrm{s}} \\
    \text{s.t.} \qquad \quad\; &\eqref{eq:sd_norm_bounded_lmi3},
        \quad~\eqref{eq:sd_norm_bounded_lmi4}, \\
         & \bm{Q}_1 = \bm{Q}_1\tp \succ \bm{0}, \quad \bm{R} = \bm{R}\tp \succ \bm{0}.
\end{split}
\end{align}
%The corresponding control gain is given by $\bm{K}=\bm{Y}\bm{Q}_1^{-1}$. 
%\end{corollary}

%
For fixed values of $\epsilon_1$ and $\epsilon_2$,~\eqref{eq:sd_norm_bounded_lmi3} becomes a linear fractional constraint of the form ${\lambda\bm{M}(\bm{s}) + \bm{N}(\bm{s}) \prec \bm{0}}$, where ${\lambda = \frac{1}{T_\mathrm{s}} \in \R}$ and ${\bm{s} \in \R^{n_s}}$ are the optimization variables, and the matrices $\bm{M}$ and $\bm{N}$ depend affinely on $\bm{s}$. Then,~\eqref{eq:sd_norm_bounded_stability_op} represents a generalized eigenvalue problem (GEVP), which is a special type of SDP that can be solved efficiently, for example, with the bisection method~\cite{Boyd.1994}.
% For fixed values of $\epsilon_1$ and $\epsilon_2$,~\eqref{eq:sd_norm_bounded_stability_op} represents a generalized eigenvalue minimization problem (GEVP), which is a special type of SDP that can be solved efficiently, for example, with the bisection method~\cite{Boyd.1994}. %\cite{Boyd.2004}. 
%We optimize over $\epsilon_1$ and $\epsilon_2$ with a logarithmic-scale grid search. 
We can simplify the additional optimization over the scalar variables~$\epsilon_1$ and $\epsilon_2$ by noting the following:
\begin{lemma} \label{lem:sd_stability_op_solution_equivalence}
    Let~$\mathcal{S}=\big(\hat{T}_\mathrm{s}$, $\bmh{Q}_1$, $\bmh{Q}_2$, $\bmh{Q}_3$, $\bmh{Z}_1$, $\bmh{Z}_2$, $\bmh{Z}_3$, $\bmh{R}$, $\bmh{Y}$, $\hat{\epsilon}_1$, $\hat{\epsilon}_2 \big)$ be an optimal solution to~\eqref{eq:sd_norm_bounded_stability_op}. Then, for any~$c > 0$, ~$\mathcal{S}'=\big(\hat{T}_\mathrm{s}$, $\frac{1}{c}\bmh{Q}_1$, $\frac{1}{c}\bmh{Q}_2$, $\frac{1}{c}\bmh{Q}_3$, $\frac{1}{c}\bmh{Z}_1$, $\frac{1}{c}\bmh{Z}_2$, $\frac{1}{c}\bmh{Z}_3$, $\frac{1}{c}\bmh{R}$, $\frac{1}{c}\bmh{Y}$, $c\hat{\epsilon}_1$, $c\hat{\epsilon}_2 \big)$ is also an optimal solution to~\eqref{eq:sd_norm_bounded_stability_op}.
\end{lemma}

\begin{proof}
    Both solutions~$\mathcal{S}$ and~$\mathcal{S}'$ yield the same objective value~$\frac{1}{\hat{T}_\mathrm{s}}$. To show that the feasibility of~$\mathcal{S}$ implies the feasibility of~$\mathcal{S}'$, we first consider the inequality constraint~\eqref{eq:sd_norm_bounded_lmi4}. 
    If~$\mathcal{S}$ is a feasible solution to~\eqref{eq:sd_norm_bounded_stability_op}, then it holds by applying the Schur complement that
    \begin{align*}
    \setlength{\arraycolsep}{2pt}
    \begin{split}
        \begin{bmatrix}
            2\bmh{Q}_1 - \bmh{R} & \bm{0} & \bmh{Y}^\top \bm{B}\tp \\
            * & \bmh{Z}_1 & \bmh{Z}_2  \\
            * & * & \bmh{Z}_3 
        \end{bmatrix} \qquad \qquad \qquad \\ 
        - \frac{1}{\hat{\epsilon}_2}
        \begin{bmatrix}
            \bm{0} & \hat{\epsilon}_2 \bmh{Y}^\top \bm{F}\tp \\
            \bm{0} & \bm{0} \\
            \bm{H}\tp & \bm{0}
        \end{bmatrix}
        \begin{bmatrix}
            \bm{0} & \bm{0} & \bm{H}\tp \\
            \hat{\epsilon}_2 \bm{F} \bmh{Y}\tp & \bm{0} & \bm{0}
        \end{bmatrix} &\succeq \bm{0} 
    \end{split}
    \end{align*}
    \begin{align*}
    \begin{split}
        \Leftrightarrow \begin{bmatrix}
            2\frac{1}{c}\bmh{Q}_1 - \frac{1}{c}\bmh{R} & \bm{0} & \frac{1}{c}\bmh{Y}^\top \bm{B}\tp \\
            * & \frac{1}{c}\bmh{Z}_1 & \frac{1}{c}\bmh{Z}_2  \\
            * & * & \frac{1}{c}\bmh{Z}_3 
        \end{bmatrix} \qquad \qquad \qquad \\ 
        - \frac{1}{c\hat{\epsilon}_2}
        \begin{bmatrix}
            \bm{0} & c\hat{\epsilon}_2 \frac{1}{c}\bmh{Y}^\top \bm{F}\tp \\
            \bm{0} & \bm{0} \\
            \bm{H}\tp & \bm{0}
        \end{bmatrix}
        \begin{bmatrix}
            \bm{0} & \bm{0} & \bm{H}\tp \\
            c\hat{\epsilon}_2 \bm{F} \frac{1}{c}\bmh{Y}\tp & \bm{0} & \bm{0}
        \end{bmatrix} &\succeq \bm{0}, 
    \end{split}
    \setlength{\arraycolsep}{3pt}
    \end{align*}
    where we have multiplied the inequality by~${\frac{1}{c} > 0}$. Consequently, $\mathcal{S}'$ also satisfies~\eqref{eq:sd_norm_bounded_lmi4}.
    We can proceed in a similar way for the constraint~\eqref{eq:sd_norm_bounded_lmi3}, which concludes the proof.
\end{proof}

The two solutions~$\mathcal{S}$ and~$\mathcal{S}'$ also yield the same stabilizing control gain $\bm{K}=\bm{Y}(\bm{Q}_1)^{-1}=\frac{1}{c}\bm{Y}\big(\frac{1}{c}\bm{Q}_1\big)^{-1}$.
As a consequence of~\cref{lem:sd_stability_op_solution_equivalence}, we can simplify the optimization problem~\eqref{eq:sd_norm_bounded_stability_op} by setting~${\epsilon_2 = \frac{1}{\epsilon_1}}$.
\begin{corollary} \label{cor:sd_norm_bounded_stability_op_simplification}
    Solving the optimization problem~\eqref{eq:sd_norm_bounded_stability_op} yields the same MCF and stabilizing control gain as solving
    \begin{align}
    \label{eq:sd_norm_bounded_stability_op_simplified}
    \setlength{\arraycolsep}{2pt}
    \begin{split}
        \min_{\substack{
                T_\mathrm{s}, \bm{Q}_1, \bm{Q}_2, \\
                \bm{Q}_3, \bm{Z}_1, \bm{Z}_2 \\
                 \bm{Z}_3, \bm{R}, \bm{Y}, \epsilon
            }}
        \; & \frac{1}{T_\mathrm{s}} \\
        \mathrm{s.t.} \quad
            &\mleft[
        \begin{array}{c c c | c c}
            \multicolumn{3}{c}{\multirow{3}{*}{\large $\bm{W}_{\bmh{A},\bmh{B}}$}} \vline & \bm{0} & \epsilon\left(\bm{Q}_1\bm{E}\tp + \bm{Y}\tp\bm{F}\tp\right) \\
            \multicolumn{3}{c}{} \vline & \bm{H} & \bm{0} \\
            \multicolumn{3}{c}{} \vline & \bm{0} & \bm{0} \\
            \hline
            * & * & * & -\epsilon \bm{I} & \bm{0} \\
            * & * & * & * & -\epsilon \bm{I}
        \end{array}
        \mright] \!\prec \bm{0}, \\
        &\!\begin{bmatrix}
            2\bm{Q}_1 - \bm{R} & \bm{0} & \bm{Y}\tp \bmh{B}\tp & \bm{0} & \epsilon \bm{Y}\tp \bm{F}\tp \\
            * & \bm{Z}_1 & \bm{Z}_2 & \bm{0} & \bm{0} \\
            * & * & \bm{Z}_3 & \bm{H} & \bm{0} \\
            * & * & * &  \epsilon \bm{I} & \bm{0} \\
            * & * & * & * & \epsilon \bm{I}
        \end{bmatrix} \succeq \bm{0}, \\
        & \bm{Q}_1 = \bm{Q}_1\tp \succ \bm{0}, \quad \bm{R} = \bm{R}\tp \succ \bm{0}, \quad \epsilon > 0.
    \end{split}
    \setlength{\arraycolsep}{4pt}
    \end{align}
\end{corollary}

\begin{proof}
    Let~$\mathcal{S}$ be an optimal solution to~\eqref{eq:sd_norm_bounded_stability_op}, as defined in~\cref{lem:sd_stability_op_solution_equivalence}.
    We can make use of~\cref{lem:sd_stability_op_solution_equivalence} and set~${c = \frac{1}{\sqrt{\hat{\epsilon}_1 \hat{\epsilon}_2}}}$. Then, $\mathcal{S}'$ corresponds to the optimal solution of~\eqref{eq:sd_norm_bounded_stability_op_simplified} and the result follows.
\end{proof}

The reformulated optimization problem~\eqref{eq:sd_norm_bounded_stability_op_simplified} has~${6n^2 + nm + n + 2}$ decision variables, which is one less than~\eqref{eq:sd_norm_bounded_stability_op}. In practice, having to perform only a scalar grid search over~$\epsilon$ for solving~\eqref{eq:sd_norm_bounded_stability_op_simplified} instead of a two-dimensional grid search over~$\epsilon_1$ and~$\epsilon_2$ for solving~\eqref{eq:sd_norm_bounded_stability_op} significantly decreases the computational demand.

\begin{remark}
    Analyzing stability via the time-delay approach~\cite{Fridman.2014} is advantageous: As long as the constraints in~\eqref{eq:sd_norm_bounded_stability_op_simplified} hold, the sampling time can be changed online in $(0,T_{\mathrm{s,max}}]$ without losing stability guarantees, for example, to react to changes in the data transmission capacity.
    % Analyzing stability via the time-delay approach~\cite{Fridman.2014} has two advantages. First, as the approach is based on the existence of an upper bound on the delay, % introduced by sampling, 
    % additional sources of delay can be accounted for by increasing the upper bound accordingly.
    % Second, as long as the constraints in~\eqref{eq:sd_norm_bounded_stability_op_simplified} hold, the sampling time can be changed online in $(0,T_{\mathrm{s,max}}]$ without losing stability guarantees, for example, to react to changes in the data transmission capacity.
\end{remark}

\subsection{Optimizing the Control Performance} \label{subsec:meth_sd_performance}
Besides guaranteeing stability, we aim to optimize the control performance for a given sampling time ${T_\mathrm{s} \in (0,T_\mathrm{s,max}]}$. As a performance measure, we use the cost~\cite{Lam.2006}
\begin{align}
    \label{eq:sd_cost}
    J = \int_{\tau_0}^{\tau_1} \bmt{x}(t_k)\tp \bm{Q}_J \bmt{x}(t_k) + \bmt{u}(t)\tp \bm{R}_J \bmt{u}(t) \mathrm{d} t_k,
\end{align}
where $\tau_1 - \tau_0 > 0$ is the optimization period, and ${\bm{Q}_J \succ \bm{0}}$, ${\bm{R}_J \succ \bm{0}}$ are weight matrices. 
To make the problem tractable, we let ${J \leq \bar{J} = \eta \int_{\tau_0}^{\tau_1} \bmt{x}(t_k)\tp \bm{Q}_1^{-1} \bm{Q}_1^{-1} \bmt{x}(t_k)\mathrm{d} t_k}$ and consider the minimization of $\eta>0$. By inserting the control law~\eqref{eq:sd_control_law_linearized}, substituting ${\bm{K}=\bm{Y}\bm{Q}_1^{-1}}$ and applying the Schur complement, we obtain that $J\leq\bar{J}$ for all $\tau_1 - \tau_0 > 0$ if
\begin{align}
    \label{sd:cost_lmi}
    \begin{bmatrix}
        -\eta \bm{I} & \bm{Q}_1 & \bm{Y}\tp \\ 
        * & -\bm{Q}_J^{-1} & \bm{0} \\
        * & * & -\bm{R}_J^{-1}
    \end{bmatrix} \preceq \bm{0}.
\end{align}
%\begin{corollary} \label{cor:performance_optimization}
Consequently, we can optimize the cost~\eqref{eq:sd_cost} while ensuring robust stability of the uncertain system~\eqref{eq:sd_norm_bounded_open_loop_system} under the control~\eqref{eq:sd_control_law_linearized} by solving the optimization problem
\begin{align}
\begin{split}
\label{eq:sd_norm_bounded_performance_op}
    \min_{\substack{
            \eta,\bm{Q}_1, \bm{Q}_2, \bm{Q}_3, \bm{Z}_1, \\
            \bm{Z}_2, \bm{Z}_3, \bm{R}, \bm{Y}, \epsilon_1, \epsilon_2
        }}
    \quad & \eta \\
    \text{s.t.} \qquad \quad\; & \bm{Q}_1 = \bm{Q}_1\tp \succ \bm{0}, \quad \bm{R} = \bm{R}\tp \succ \bm{0}, \\
        &~\eqref{eq:sd_norm_bounded_lmi3},\quad~\eqref{eq:sd_norm_bounded_lmi4}, \quad~\eqref{sd:cost_lmi},
\end{split}
\end{align}
for a given sampling time $T_\mathrm{s} \in \left(0,T_{\mathrm{s,max}}\right]$. 
Note that~\eqref{eq:sd_norm_bounded_performance_op} has a GEVP structure similar to~\eqref{eq:sd_norm_bounded_stability_op}. However, \eqref{eq:sd_norm_bounded_performance_op} does not allow for a simplification of the optimization over the scalar decision variables~$\epsilon_1$ and~$\epsilon_2$ since~\cref{lem:sd_stability_op_solution_equivalence} does not apply due to the inequality constraint~\eqref{sd:cost_lmi}.

In this section, we have derived a method to calculate the MCF $f_\mathrm{c,min}$ required for
robust control of~\eqref{eq:prob_open_loop_system} based on the uncertain linearized GP dynamics model~\eqref{eq:sd_norm_bounded_open_loop_system}. We can also compute a robustly stabilizing and optimal state-feedback controller~\eqref{eq:sd_control_law_linearized} for a given control frequency~${f_\mathrm{c} \geq f_{\mathrm{c,min}}}$.
Next, we evaluate our approach in simulation and study the tradeoff between data and the control frequency.

\section{Evaluation} \label{sec:eval}
We consider a quadrotor flying in the vertical plane with position~$(x,z)$ and pitch angle~$\theta$. The system is simulated with the continuous-time dynamics~\cite{Yuan.2022}
\begin{align}
\begin{split}
    \label{eq:eval_quadrotor_dynamics}
    m\ddot{x} &= -(T_1 + T_2)\sin{(\theta)} \\
    m\ddot{z} &=  (T_1 + T_2)\cos{(\theta)} - mg \\
    I_{yy}\ddot{\theta} &= (T_1 -  T_2)d,
\end{split}
\end{align}
where ${m=0.1\,\si{kg}}$ is the mass, ${(T_1, T_2)}$ are the motor thrusts, ${g=9.81\,\frac{\si{m}}{\si{s}^2}}$ is the gravitational acceleration, ${d=0.1\,\si{m}}$ is the length of the effective moment arm of the propellers, and ${I_{yy}=\frac{1}{12}md^2}$ is the inertia about the $y$-axis. 
By defining ${\bm{x}=\big[x,\,\dot{x},\,z,\,\dot{z},\,\theta,\,\dot{\theta}\big]\tp}$ and ${\bm{u}=\big[T_1,\,T_2\big]\tp}$, we can express~\eqref{eq:eval_quadrotor_dynamics} in the general form~\eqref{eq:prob_open_loop_system}. 
We assume no prior knowledge about the dynamics, i.e., $\bm{f}\equiv\bm{0}$.
%to illustrate the broad applicability of the GP learning approach.
%set $\bm{f}\equiv\bm{0}$. 
We set the noise variance to ${\bm{\Sigma}_\mathrm{n} = \mathrm{Diag}([0.1^2,\dots,0.1^2])}$, and draw $N$~training inputs uniformly from the set ${\mathcal{Z} = \left\{(\bm{x}, \bm{u}) \in \R^8\left|\, \bm{\underline{x}} \leq \bm{x} \leq \bmb{x},\,\bm{\underline{u}} \leq \bm{u} \leq \bmb{u} \right. \right\}}$, where ${\bm{\underline{x}} = \big[0,-5\,\frac{\si{m}}{\si{s}},0,-5\,\frac{\si{m}}{\si{s}},-\frac{\pi}{2}\,\si{rad},-5\,\frac{\si{rad}}{\si{s}}\big]\tp}$, 
${\bmb{x} = \big[2\,\si{m},5\,\frac{\si{m}}{\si{s}},2\,\si{m},5\,\frac{\si{m}}{\si{s}},\frac{\pi}{2}\,\si{rad},5\,\frac{\si{rad}}{\si{s}}\big]\tp}$, ${\bm{\underline{u}} = \bm{0}}$ and ${\bmb{u} = \big[2\,\si{N},2\,\si{N}\big]\tp}$.
We set the threshold in~\cref{lem:linearized_dynamics} to~${p^n=0.99^6\approx0.94}$. %, which yields $\sqrt{\chi_{n_z}^2(p)}=4.5$ 
The SDPs~\eqref{eq:sd_norm_bounded_stability_op_simplified} and~\eqref{eq:sd_norm_bounded_performance_op} are solved using YALMIP~\cite{Lofberg.2004} and MOSEK~\cite{Mosek.2022}. Additionally, we define the grid~${\mathcal{G}=\{10^{-3},10^{-2.7},\dots,10^{3}\}}$ and optimize over~${\epsilon \in \mathcal{G}}$ in~\eqref{eq:sd_norm_bounded_stability_op_simplified} and over~${(\epsilon_1,\epsilon_2) \in \mathcal{G} \times \mathcal{G}}$ in~\eqref{eq:sd_norm_bounded_performance_op}.
%The computation time is about 30 \si{s} for solving~\eqref{eq:sd_norm_bounded_stability_op} and 100 \si{s} for solving~\eqref{eq:sd_norm_bounded_performance_op} on a laptop with an Intel(R) Core(TM) i7-12800HX processor and 32 GB of RAM. 

\pgfplotsset{compat=1.8}
\begin{figure}[tb!]
\centering
\setlength{\fwidth}{5.5cm}
\setlength{\fheight}{2cm}
\definecolor{mycolor1}{RGB}{31,119,180}
\definecolor{mycolor2}{RGB}{214,39,40}
\begin{tikzpicture}

\begin{axis}[%
width=\fwidth,
height=\fheight,
ticklabel style = {font=\small},
label style = {font=\small},
scale only axis,
axis y line*=left,
xmin=125,
xmax=1025,
xlabel={Number of data points $N$},
xlabel style={yshift=4pt},
ymin=8,
ymax=34.5,
ylabel style={align=center, xshift=-2pt},
ylabel={$f_{\mathrm{c,min}}$ in \si{Hz} \textcolor{mycolor1}{\rule[0.5ex]{0.3cm}{1pt}}},
axis background/.style={fill=white}
]
\addplot [color=mycolor1, forget plot, thick]
 plot [error bars/.cd, y dir=both, y explicit, error bar style={line width=0.5pt}, error mark options={line width=0.5pt, mark size=3.0pt, rotate=90}]
 table[row sep=crcr, y error plus index=2, y error minus index=3]{%
100	nan	nan	nan\\
150	31.3558349609375	2.57649012308747	2.57649012308747\\
200	25.1954055786133	4.35777815211613	4.35777815211613\\
250	18.5747634887695	2.96231039990762	2.96231039990762\\
300	16.5877937316895	2.29763411883318	2.29763411883318\\
350	14.1695198059082	2.24777441251922	2.24777441251922\\
400	13.5322738647461	1.75136117195063	1.75136117195063\\
450	12.0191940307617	1.96972510554891	1.96972510554891\\
500	12.2783363342285	1.5802878925557	1.5802878925557\\
550	11.6636932373047	1.81576719728043	1.81576719728043\\
600	11.1952407836914	1.19121745204085	1.19121745204085\\
650	10.8508255004883	1.11623381196154	1.11623381196154\\
700	10.7335830688477	0.958345289846558	0.958345289846558\\
750	10.3948097229004	1.14124743738174	1.14124743738174\\
800	10.1885543823242	1.05120641105878	1.05120641105878\\
850	9.96351928710937	1.26135576096231	1.26135576096231\\
900	9.88223571777344	1.07121419465824	1.07121419465824\\
950	9.39556732177734	0.926335203771018	0.926335203771018\\
1000	9.38313369750977	1.0664377963072	1.0664377963072\\
};
\label{plots:fcmin}
\end{axis}

\begin{axis}[
width=\fwidth,
height=\fheight,
ticklabel style = {font=\small},
label style = {font=\small},
scale only axis,
axis y line*=right,
xmin=125,
xmax=1025,
axis x line=none,
ymin=-0.05, ymax=1.05,
ylabel style={align=center, xshift =-10pt},
ylabel={Share of datasets for \\ which~\eqref{eq:sd_norm_bounded_stability_op_simplified} is feas. \ref{plots:success_rate}}
]
\addplot[mycolor2, mark=o, only marks, thick]
 table[row sep=crcr]{%
    100 0 \\
    150 0.2 \\
    200 0.5 \\
    250 1 \\
    300 1 \\
    350 1 \\
    400 1 \\
    450 1 \\
    500 1 \\
    550 1 \\
    600 1 \\    
    650 1 \\ 
    700 1 \\ 
    750 1 \\ 
    800 1 \\ 
    850 1 \\ 
    900 1 \\ 
    950 1 \\ 
    1000 1 \\ 
}; 
\label{plots:success_rate}
\end{axis}
\end{tikzpicture}%
\vspace{-0.1cm}

\setlength{\belowcaptionskip}{-2pt}

\caption{Minimum control frequency required to ensure robust stability %of the learned system 
for different amounts of randomly drawn 
training data. The error bars represent $\pm$ one standard deviation. The circles show the proportion of datasets for which a stabilizing controller can be found by solving~\eqref{eq:sd_norm_bounded_stability_op_simplified}.}
    \label{fig:2dquad_fc_N}
\end{figure}
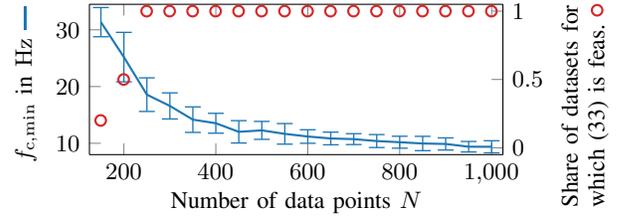
\cref{fig:2dquad_fc_N} shows the MCF computed by solving~\eqref{eq:sd_norm_bounded_stability_op_simplified} for randomly drawn training sets of increasing size $N$, where ten different datasets are drawn for each value of $N$. We also provide the proportion of datasets for which~\eqref{eq:sd_norm_bounded_stability_op_simplified} is feasible. 
We observe that a certain amount of data is required to stabilize the system and that the control frequency can be reduced by two-thirds when more data is available.
\definecolor{mycolor1}{RGB}{31,119,180}
\definecolor{mycolor2}{RGB}{214,39,40}
\definecolor{mycolor3}{RGB}{44,160,44}
\definecolor{mycolor4}{RGB}{255,127,14}
\begin{figure}
\centering    
\includegraphics[width=0.47\textwidth, natwidth=660, natheight=443]{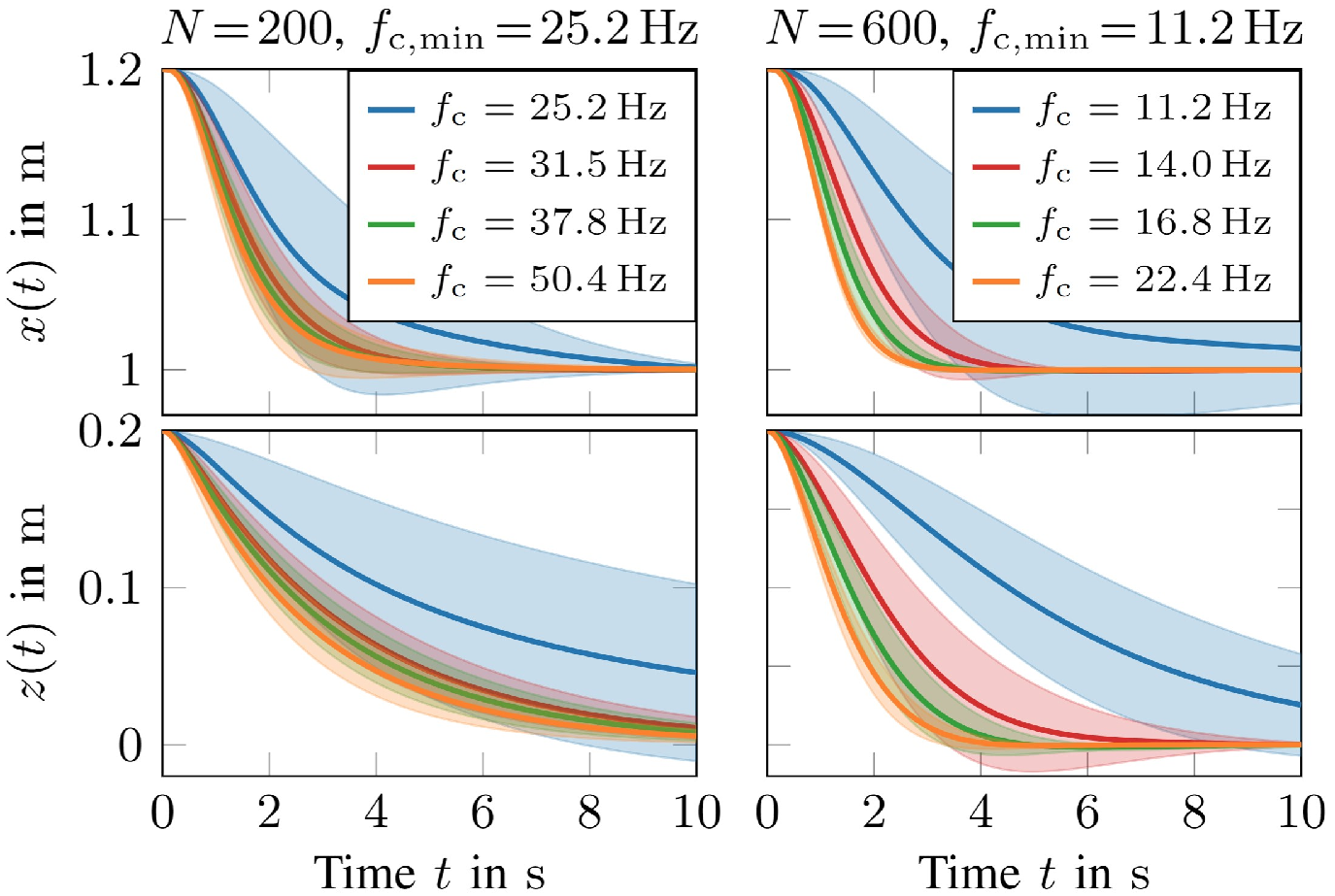}
    \vspace{-0.1cm}
    \setlength{\belowcaptionskip}{-2pt}
    \caption{Quadrotor trajectories for different control frequencies ${f_\mathrm{c}=\xi f_\mathrm{c,min}}$, where ${\xi\in \{{\color{mycolor1} 1},{\color{mycolor2} 1.25},{\color{mycolor3} 1.5},{\color{mycolor4} 2}\}}$, and different amounts of training data. The shaded areas represent $\pm$ one standard deviation. 
    Convergence to the setpoint significantly improves, and variance reduces if the control frequency is increased from its minimum value $f_\mathrm{c}=f_\mathrm{c,min}$.}
    \label{fig:impact_of_xi}
\end{figure}  % For ARXIV

We also investigate the impact of the control frequency and model uncertainty on performance. For this, we set the desired operating point to ${\bm{x}_\mathrm{e} = [1\,\si{m},0,0,0,0,0]\tp}$, ${\bm{u}_\mathrm{e} = [0.4905\,\si{N},0.4905\,\si{N}]\tp}$, the initial state to $\bm{x}(0) = [1.2\,\si{m},0,0.2\,\si{m},0,0,0]\tp$ and the weight matrices in~\eqref{eq:sd_cost} to ${\bm{Q}_J = \mathrm{diag}([100,1,100,1,100,1])}$, ${\bm{R}_J = 0.01 \bm{I}_2}$.
%to encourage aggressive control behavior.
As discussed in~\cref{subsec:meth_sd_stability}, robust stability of the linearized system is guaranteed for all control frequencies $f_\mathrm{c} =\xi f_\mathrm{c,min}$ with $\xi \geq 1$. We evaluate ${\xi\in \{1,1.25,1.5,2\}}$ with ten randomly drawn training sets each and compute the optimized controller by solving~\eqref{eq:sd_norm_bounded_performance_op}. \cref{fig:impact_of_xi} shows the %obtained quadrotor position
quadrotor trajectories obtained from simulating the system for $10\,\si{s}$ with mean and standard deviation for ${N\in\{200,600\}}$. We observe the transient behavior improves significantly, and the variance reduces if the control frequency is increased from its minimum value. 
\input{figures/cost_scatter}

For a systematic analysis of performance, we optimize for many frequencies ${f_\mathrm{c} \in [10,30]\,\si{Hz}}$ and simulate five different initial conditions.
\cref{fig:cost_scatter} shows the average LQR cost for a horizon of $10\,\si{s}$ with the weight matrices $\bm{Q}_J$ and $\bm{R}_J$ and the contour lines of the cost.
We observe a tradeoff between data and control frequency: As model uncertainty decreases due to more data, similar performance is achieved at a lower control frequency. 
On the other hand, for example, if we increase $f_\mathrm{c}$ by~$33\%$ from~${18\,\si{Hz}}$ to~${24\,\si{Hz}}$, only half the amount of data (${N=400}$ instead of~${N=800}$) is needed to get the same cost~$J=5.6$.
Increasing the control frequency for a given amount of data reduces the cost, for example, by ~$42\%$ when increasing~$f_\mathrm{c}$ by~$29\%$ from~${14\,\si{Hz}}$ to~${18\,\si{Hz}}$ for~${N=350}$.
Furthermore, it is evident from the contour lines' shape that the sensitivity of the performance with respect to the control frequency increases with the size of the training data set.
%The achieved performance remains consistent and even slightly improves for an increasing amount of data, although the control frequency is reduced.
%, except for ${f_\mathrm{c}=f_\mathrm{c, min}}$, which results in significantly higher cost values. The achieved performance remains consistent and even slightly improves for an increasing amount of data, although the control frequency is reduced. %Furthermore, for ${\xi\in\{1.5,2\}}$, improving the model avoids large outliers in the cost.
% \input{figures/cost_boxplots_matlab2tikz}

\section{Discussion} \label{sec:disc}
%Treating the considered problem with the input-delay approach to sampled-data systems~\cite{Fridman.2014} offers two advantages. 
%First, delays in the data transmission can easily be taken into account since they increase the upper bound on the delay $\tau(t)$.
%\begin{remark}
    %We assume the availability of time derivatives of the state perturbed by Gaussian noise in \cref{prob:eq_dataset_assumption}. 
%\end{remark}

% Upper bound on sampling time: We can also consider other sources of delay.

\cref{fig:impact_of_xi} and~\cref{fig:cost_scatter} demonstrate that a higher control frequency improves performance and reduces variance. Considering this, the MCF provides a lower bound that allows us to safely reduce the control frequency, for example, to save computation or communication resources.

As illustrated in~\cref{fig:2dquad_fc_N} and~\cref{fig:cost_scatter}, the amount of data required for stability or achieving a specific performance depends on the frequency at which we can run the controller. A slight increase in control frequency can compensate for a significant lack of data, as is typical for physical systems such as robots, where data collection is expensive. 

Stable regulation of the quadrotor is achieved for all simulated initial conditions, even when operating at the MCF, as shown in~\cref{fig:impact_of_xi}.
This is remarkable, as the computation of the MCF via~\eqref{eq:sd_norm_bounded_stability_op_simplified} is based on the linearized dynamics~\eqref{eq:sd_norm_bounded_open_loop_system} and only considers the uncertainty corresponding to the GP variance, not the error due to linearizing the true nonlinear system~\eqref{eq:prob_open_loop_system}. One reason is that the stability conditions in~\cref{the:sd_norm_bounded_step2} are sufficient but not necessary conditions, and thus, the MCF includes some conservatism that may compensate for the linearization error, as is the case for our example.

\section{Conclusion} \label{sec:conc}
This work considers the control frequency as a design parameter for learning-based control of uncertain systems. 
To this end, we combine learning a continuous-time dynamics model using GPs with robust sampled-data control. This enables us to control the system at different frequencies without re-learning the model and to study the role of both the control frequency and the amount of data.
%The approach allows us to study the tradeoff between data and control frequency. 
%It allows us to study the tradeoff between data and control frequency: 
%Using our proposed approach, we show that the control frequency is as important as the number of data points in terms of stability and performance. 
We show that there is a tradeoff between the two design parameters in terms of stability and performance: 
Increasing the control frequency can make up for a lack of data and vice versa.

\bibliographystyle{IEEEtran}
\bibliography{ref}

\end{document}